%% file: many.tex
\documentclass[conference,a4paper]{IEEEtran}
\pdfoutput=1

\input{preamble}

\usepackage{cite}

\title{Optimal Non-Uniform Mapping\\for Probabilistic Shaping}

\IEEEoverridecommandlockouts

\author{\IEEEauthorblockN{Georg B\"ocherer}
\IEEEauthorblockA{Institute for Communications Engineering\\Technische Universit\"at M\"unchen, Germany\\
Email: \texttt{georg.boecherer@tum.de}}
\thanks{This work was supported by the German Ministry of Education and Research in the framework of an Alexander von Humboldt Professorship.}
}

\newcommand{\snr}{\mathsf{snr}}

\begin{document}

\maketitle

\input{abstract}

\input{introduction}

\input{problem}

\input{algorithm}

\input{awgn}

\input{numerical}



\bibliographystyle{IEEEtran}
\bibliography{IEEEabrv,confs-jrnls,references}

\end{document}

%% file: preamble.tex
\usepackage{amsmath}
\usepackage{amsthm}
\usepackage{amssymb}
\usepackage{bm}
\usepackage{xspace}
\usepackage{xcolor}
\usepackage{graphicx}
\usepackage{url}
\usepackage{framed}
\usepackage{float}
\usepackage{rotating}
\usepackage{verbatim}
\usepackage{listings}
\usepackage{lscape}

\usepackage{pgfplots}
\usepackage{pgf}
\usepackage{tikz}
\usetikzlibrary{arrows,shapes.misc,chains,scopes}
\usepackage{pgfplotstable}
\usepackage{booktabs}
\usepackage{colortbl}

\newcommand{\executeiffilenewer}[3]{%
\ifnum\pdfstrcmp{\pdffilemoddate{#1}}%
{\pdffilemoddate{#2}}>0%
{\immediate\write18{#3}}\fi%
}

\graphicspath{{figures/}}

\theoremstyle{plain}
\newtheorem{proposition}{Proposition}

\newtheorem{lemma}{Lemma}

\newcounter{algocount}
\newcounter{examplecount}

\newenvironment{algorithm}[1][]{\refstepcounter{algocount}\setlength{\parindent}{0.5cm}\begin{trivlist}\item \textbf{Algorithm \thealgocount.}#1\\[-0.2cm]\rule{\columnwidth}{1pt}}{\\[-0.2cm]\rule{\columnwidth}{1pt}\end{trivlist}}

\newcommand{\veczero}{\boldsymbol{0}}

\newcommand{\vecd}{\boldsymbol{d}}

\newcommand{\vecc}{\boldsymbol{c}}

\newcommand{\vecp}{\boldsymbol{p}}

\newcommand{\vect}{\boldsymbol{t}}

\newcommand{\setz}{\ensuremath{\mathbf{Z}}\xspace}

\newcommand{\bpm}{\begin{pmatrix}}
\newcommand{\epm}{\end{pmatrix}}
\newcommand{\bbm}{\begin{bmatrix}}
\newcommand{\ebm}{\end{bmatrix}}

\DeclareMathOperator*{\argmin}{argmin}

\DeclareMathOperator{\expop}{\mathbb{E}}

\DeclareMathOperator{\miop}{\mathbb{I}}
\DeclareMathOperator{\kl}{\mathbb{D}}

\DeclareMathOperator*{\minimize}{minimize}
\DeclareMathOperator*{\maximize}{maximize}
\DeclareMathOperator*{\st}{subject\;to}

%% file: abstract.tex
\begin{abstract}
The construction of optimal non-uniform mappings for discrete input memoryless channels (DIMCs) is investigated. An efficient algorithm to find optimal mappings is proposed and the rate by which a target distribution is approached is investigated. The results are applied to non-uniform mappings for additive white Gaussian noise (AWGN) channels with finite signal constellations. The mappings found by the proposed methods outperform those obtained via a central limit theorem approach as suggested in the literature.
\end{abstract}

%% file: introduction.tex
\section{Introduction}

The capacity of a discrete input memoryless channel (DIMC) is given by the maximum mutual information between the channel input and the channel output, where the maximum is taken over all permitted input probability mass functions (pmf). For a digital communication system to operate close to capacity, the pmf of the channel input symbols should resemble the capacity-achieving pmf. Unequal transition probabilities between input and output symbols, cost constraints, or input symbols of unequal durations can lead to non-uniform capacity-achieving input pmfs \cite{bocherer2012capacity}. Techniques to achieve non-uniform pmfs go under the name probabilistic shaping. Recently, \c{S}a\c{s}o\u{g}lu \emph{et al} \cite{sasoglu2009polarization} constructed polar codes that achieve the \emph{symmetric} capacity for arbitrary DIMCs, i.e., the maximum rate for \emph{uniform} input pmfs. This raises the question: can these codes achieve the true capacity?

One possibility to address this problem is by wrapping the channel by a super-channel that permits a uniform input. Gallager proposed in \cite[p. 208]{Gallager1968} to use a non-uniform mapping from $M$ symbols to the channel input alphabet to realize such a super-channel. An example of such a mapping is displayed in Fig.~\ref{fig:num}. This mapping transforms a uniform distribution over $M=4$ symbols into the non-uniform pmf $d_1=1/4$, $d_2=3/4$, $d_3=0$. This non-uniform mapping approach is briefly discussed in \cite[Sec. III.D]{sasoglu2009polarization}. However, if the mapping requires a very large $M$, then it may not be practical since coding must be done over the $M$ symbols and therefore the coding complexity increases with $M$, see \cite{abbe2011polar}. This observation motivates looking for efficient non-uniform mappings. 

For a uniform distribution over $M$ symbols, each mapping generates an $M$-type pmf, i.e., a pmf where each symbol probability can be written as $c/M$ for some non-negative integer $c$. Conversely, for each $M$-type pmf $\vecd$ there is a mapping that generates it. Note that the mapping is in general many-to-one and not necessarily onto. The mapping in Fig.~\ref{fig:num} is an example. We focus on the construction of $M$-type pmfs; the corresponding mapping is easily obtained.
\begin{figure}
\centering
\footnotesize
\def\svgwidth{0.5\columnwidth}
\executeiffilenewer{figures/num.svg}{figures/num.pdf}%
{inkscape -z -D --file=figures/num.svg %
--export-pdf=figures/num.pdf --export-latex}%
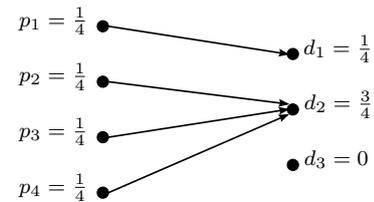%

\caption{The displayed mapping transforms a uniform distribution over $4$ symbols into the non-uniform pmf $d_1=1/4$, $d_2=3/4$, $d_3=0$.}
\label{fig:num}
\vspace{-0.4cm}
\end{figure}

We ask the following two questions:
\begin{itemize}
\item[Q1] When we increase $M$, how fast can an $M$-type pmf converge to the target pmf?
\item[Q2] For a finite $M$, how can we find the $M$-type pmf that ``optimally'' approximates the target pmf?
\end{itemize}
In \cite[Sec. IV.B]{abbe2011polar}, Abbe and Barron consider question Q1 for the \emph{additive white Gaussian noise} (AWGN) channel. For $M=2^m$, they suggest to use the binomial coefficients divided by $M$ as probabilities for an $(m+1)$-PAM constellation. They call their method the \emph{central limit theorem} (CLT) approach and they show that the gap to the AWGN capacity $0.5\log(1+\mathsf{snr})$ scales as $1/\log(M)$. Schreckenbach proposed in \cite{Schreckenbach2007} a greedy algorithm to construct an $M$-type pmf based on a target pmf. However, the author does not address questions Q1 and Q2.  

In this work, we use the relative entropy $\kl(\vecd\Vert\vect)$ as a measure for how good $\vecd$ approximates the target pmf $\vect$. Our motivation is that this measure is an upper bound for the loss of mutual information when a pmf $\vecd$ different from the capacity achieving pmf $\vect$ is used \cite[Sec.~3.4.3]{bocherer2012capacity}. Regarding question Q1, we show that the relative entropy has an upper bound proportional to $1/M$. For question Q2, we propose an efficient algorithm that finds the $M$-type pmf that minimizes $\kl(\vecd\Vert\vect)$. The complexity of our algorithm is $\mathcal{O}(Mn)$ where $n$ is the number of entries of the target pmf $\vect$. 

This paper is organized as follows. In Sec.~\ref{sec:quantization}, we state the problem. We derive a convergence rate bound in Sec.~\ref{sec:convergence}. Sec.~\ref{sec:optimal} gives an algorithm to find optimal $M$-type approximations. In Sec.~\ref{sec:awgn} and Sec.~\ref{sec:numerical}, we apply our methods to the AWGN channel and provide numerical results. The mappings found by the proposed methods outperform those obtained via the CLT approach as suggested in \cite[Sec. IV.B]{abbe2011polar}.

%% file: figures/num.pdf_tex

\begingroup
  \makeatletter
  \providecommand\color[2][]{%
    \errmessage{(Inkscape) Color is used for the text in Inkscape, but the package 'color.sty' is not loaded}
    \renewcommand\color[2][]{}%
  }
  \providecommand\transparent[1]{%
    \errmessage{(Inkscape) Transparency is used (non-zero) for the text in Inkscape, but the package 'transparent.sty' is not loaded}
    \renewcommand\transparent[1]{}%
  }
  \providecommand\rotatebox[2]{#2}
  \ifx\svgwidth\undefined
    \setlength{\unitlength}{481.715625pt}
  \else
    \setlength{\unitlength}{\svgwidth}
  \fi
  \global\let\svgwidth\undefined
  \makeatother
  \begin{picture}(1,0.53475533)%
    \put(0,0){\includegraphics[width=\unitlength]{num.pdf}}%
    \put(-0.00071035,0.51650858){\color[rgb]{0,0,0}\makebox(0,0)[lb]{\smash{$p_1=\frac{1}{4}$}}}%
    \put(-0.00071035,0.3504355){\color[rgb]{0,0,0}\makebox(0,0)[lb]{\smash{$p_2=\frac{1}{4}$}}}%
    \put(-0.00071035,0.18436241){\color[rgb]{0,0,0}\makebox(0,0)[lb]{\smash{$p_3=\frac{1}{4}$}}}%
    \put(-0.00071035,0.01828933){\color[rgb]{0,0,0}\makebox(0,0)[lb]{\smash{$p_4=\frac{1}{4}$}}}%
    \put(0.84460165,0.43181131){\color[rgb]{0,0,0}\makebox(0,0)[lb]{\smash{$d_1=\frac{1}{4}$}}}%
    \put(0.84626238,0.26739895){\color[rgb]{0,0,0}\makebox(0,0)[lb]{\smash{$d_2=\frac{3}{4}$}}}%
    \put(0.84626238,0.10132587){\color[rgb]{0,0,0}\makebox(0,0)[lb]{\smash{$d_3=0$}}}%
  \end{picture}%
\endgroup

%% file: problem.tex
\section{Problem Statement}
\label{sec:quantization}
\begin{figure}
\centering
\footnotesize
\def\svgwidth{1.0\columnwidth}
\executeiffilenewer{figures/quantization.svg}{figures/quantization.pdf}%
{inkscape -z -D --file=figures/quantization.svg %
--export-pdf=figures/quantization.pdf --export-latex}%
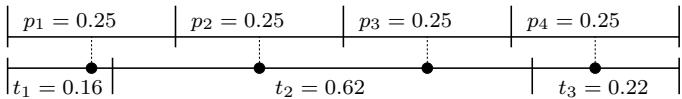%

\caption{Example for quantization as defined in \eqref{eq:quantization}. First, the left-open interval $(0,1]$ is partitioned into $M=4$ left-open uniform intervals of length $1/M$. Second, the interval $(0,1]$ is partitioned into left-open intervals whose lengths are the probabilities of the target pmf. Finally, the approximation of $t_i$ is determined by the number of uniform intervals whose middle points lie within the interval that corresponds to $t_i$. Thus, the quantization of $\vect^T=(0.16,\;0.62,\;0.22)^T$ is $\vecd^T=\frac{1}{4}(1,\;2,\;1)^T$.
}
\label{fig:quantization}
\end{figure}
\subsection{Quantization}
The \emph{cumulative distribution function} (cdf) $\boldsymbol{T}$ for the target pmf $\vect$ is defined  by
\begin{align}
T_i = \sum_{k=1}^i t_k,\quad i=1,\dotsc,n.
\end{align}
The $i$th entry of the $M$-type approximation by quantizing $\vect$ is given by
\begin{align}
d_i = \frac{1}{M}\cdot\left|\left\{\ell\in\setz\colon T_{i-1}<\frac{\ell-\frac{1}{2}}{M}\leq T_i\right\}\right|\label{eq:quantization}
\end{align}
where $\setz$ denotes the set of integers and $|\cdot|$ the cardinality of a set. We define $T_0=0$. An illustrating example is displayed in Fig.~\ref{fig:quantization}. Note that if $t_i=0$, then $T_{i-1}=T_i$, which implies that the set on the right hand side of \eqref{eq:quantization} is empty. Consequently, we have
\begin{align}
t_i=0\quad\Rightarrow\quad d_i = 0.\label{eq:probZero}
\end{align}
We make use of \eqref{eq:probZero} later. For each $i$, $d_i$ is bounded by
\begin{align}
t_i-\frac{1}{M}\leq d_i\leq t_i+\frac{1}{M}\label{eq:quantBounds}
\end{align}
which implies
\begin{align}
|t_i-d_i|\leq\frac{1}{M}.
\end{align}
This observation immediately gives the following proposition.
\begin{proposition}\label{prop:general}
Let $f$ be a continuous function from the set of pmfs with $n$ entries to the set of real numbers. Then a target pmf $\vect$ can be approximated arbitrarily well by an $M$-type pmf in the sense that for any $\epsilon>0$, there is an $M_0$, such that for all $M>M_0$ we have $|f(\vect)-f(\vecd_M)|<\epsilon$ where $\vecd_M$ is the $M$-type pmf found by quantizing $\vect$ according to \eqref{eq:quantization}.
\end{proposition}
Prop.~\ref{prop:general} applies to any continuous function defined on the probability simplex. In particular, it applies to information measures such as entropy and mutual information, which are continuous functions of the channel input pmf.
\subsection{Minimizing Relative Entropy}
Prop.~\ref{prop:general} is a qualitative result. It tells us that we can approximate a target pmf as close as desired, but it does not give the speed of convergence when $M$ increases, nor how to optimally quantize for a finite $M$. To get such results, we must specify a measure of approximation. One useful measure is the gap to capacity that results from using an $M$-type pmf instead of the capacity-achieving pmf. In \cite[Sec. IV.B]{abbe2011polar}, the authors derived a bound on this gap for AWGN channels when using $M$-type pmfs. However, the derivation depends heavily on having Gaussian noise. Getting similar results for general DIMCs seems difficult.

The relative entropy $\kl(\vecd\Vert\vect)$ of the channel input pmf $\vecd$ and the capacity-achieving pmf $\vect$ is an upper-bound on the gap to capacity that results from using $\vecd$ \cite[Sec. 3.4.3]{bocherer2012capacity}. Relative entropy is simpler to analyze than the exact gap to capacity since the (possibly complicated) structure of the channel enters only via the capacity-achieving pmf. We will therefore address question Q1 (rate of convergence) and question Q2 (optimal $M$-type pmf) with respect to $\kl(\vecd\Vert\vect)$.
\section{Convergence Rate}
\label{sec:convergence}
The relative entropy achieved by the $M$-type pmf $\vecd$ obtained by quantizing $\vect$  according to \eqref{eq:quantization} is bounded as
\begin{align}
\kl(\vecd\Vert\vect)&=\sum_{i\colon d_i>0} d_i\log\frac{d_i}{t_i}\\
&\overset{\text{(a)}}{\leq} \sum_{i\colon d_i>0} d_i\log\frac{t_i+\frac{1}{M}}{t_i}\\
&= \sum_{i\colon d_i>0} d_i\log\Bigl(1+\frac{1}{Mt_i}\Bigr)\\
&\overset{\text{(b)}}{\leq} \sum_{i\colon d_i>0} d_i\frac{1}{Mt_i}\\
&\overset{\text{(c)}}{\leq} \sum_{i\colon d_i>0} d_i\frac{1}{\displaystyle M\min_{j\colon t_j>0}t_j}\\
&= \frac{1}{\displaystyle \min_{j\colon t_j>0}t_j}\cdot\frac{1}{M}.
\end{align}
where (a) follows by \eqref{eq:quantBounds}, (b) follows by \text{$\log(1+x)\leq x$}, and where (c) follows by \eqref{eq:probZero}. Thus we have the following result.
\begin{proposition}\label{prop:rate}
For each target pmf $\vect$ there exists a constant $T>0$ such that
\begin{align}
\kl(\vecd_M\Vert\vect)\leq T/M,\quad\forall M\geq 1
\end{align}
where $\vecd_M$ is the $M$-type pmf obtained by quantizing $\vect$ according to \eqref{eq:quantization}.
\end{proposition}

%% file: figures/quantization.pdf_tex

\begingroup
  \makeatletter
  \providecommand\color[2][]{%
    \errmessage{(Inkscape) Color is used for the text in Inkscape, but the package 'color.sty' is not loaded}
    \renewcommand\color[2][]{}%
  }
  \providecommand\transparent[1]{%
    \errmessage{(Inkscape) Transparency is used (non-zero) for the text in Inkscape, but the package 'transparent.sty' is not loaded}
    \renewcommand\transparent[1]{}%
  }
  \providecommand\rotatebox[2]{#2}
  \ifx\svgwidth\undefined
    \setlength{\unitlength}{1026.4pt}
  \else
    \setlength{\unitlength}{\svgwidth}
  \fi
  \global\let\svgwidth\undefined
  \makeatother
  \begin{picture}(1,0.1343507)%
    \put(0,0){\includegraphics[width=\unitlength]{quantization.pdf}}%
    \put(0.00896337,0.00185874){\color[rgb]{0,0,0}\makebox(0,0)[lb]{\smash{$t_1 = 0.16$}}}%
    \put(0.39867498,0.00185874){\color[rgb]{0,0,0}\makebox(0,0)[lb]{\smash{$t_2 = 0.62$}}}%
    \put(0.81956352,0.00185874){\color[rgb]{0,0,0}\makebox(0,0)[lb]{\smash{$t_3=0.22$}}}%
    \put(0.02455183,0.10318376){\color[rgb]{0,0,0}\makebox(0,0)[lb]{\smash{$p_1 = 0.25$}}}%
    \put(0.27396726,0.10318376){\color[rgb]{0,0,0}\makebox(0,0)[lb]{\smash{$p_2 = 0.25$}}}%
    \put(0.5233827,0.10318376){\color[rgb]{0,0,0}\makebox(0,0)[lb]{\smash{$p_3 = 0.25$}}}%
    \put(0.77279813,0.10318376){\color[rgb]{0,0,0}\makebox(0,0)[lb]{\smash{$p_4 = 0.25$}}}%
  \end{picture}%
\endgroup

%% file: algorithm.tex
\section{Optimal $M$-type pmf}
\label{sec:optimal}
Consider a target pmf $\vect$ with $n$ entries and a number $M$. We wish to solve the optimization problem
\begin{align}
\begin{split}
\minimize_{\vecd}\quad&\kl(\vecd\Vert\vect)\\
\st\quad&\vecd\text{ is $M$-type}.
\end{split}\label{prob:orig}
\end{align}
\subsection{Equivalent Problem}
Recall that each entry $d_i$ of an $M$-type pmf can be written as $d_i=c_i/M$ for some non-negative integer $c_i$. We write the objective function of problem \eqref{prob:orig} as
\begin{align}
\kl(\vecd\Vert\vect)&=\sum_{i\colon c_i>0} \frac{c_i}{M}\log\frac{\frac{c_i}{M}}{t_i}\\
&=\frac{1}{M}\Bigl(\sum_{i\colon c_i>0} c_i\log\frac{c_i}{t_i}\Bigr)-\log M.
\end{align}
We conclude that Problem~\eqref{prob:orig} is equivalent to
\begin{align}
\begin{split}
\minimize_{c_1,\dotsc,c_n}\quad&\sum_{i\colon c_i>0} c_i\log\frac{c_i}{t_i}\\
\st\quad&c_i\in\{0,1,2,\dotsc,M\},\quad i=1,\dotsc,n\\
&\sum_{i=1}^n c_i = M.
\end{split}\label{prob:mod}
\end{align}
If $\vecc^*$ is a solution of Problem~\eqref{prob:mod}, then $\vecd^*=\vecc^*\cdot 1/M$ is a solution of Problem \eqref{prob:orig}. We call a vector $\vecc$ that fulfills the constraints of problem \eqref{prob:mod} an \emph{allocation}. 

\subsection{Algorithm}
To solve problem~\eqref{prob:mod}, we write the objective function as a telescoping sum
\begin{align}
\sum_{i\colon c_i>0} c_i\log\frac{c_i}{t_i}&=\sum_{i=1}^{n}\sum_{k=1}^{c_i}\Bigl[\underbrace{k\log\frac{k}{t_i}-(k-1)\log\frac{k-1}{t_i}}_{=:\Delta_i(k)}\Bigr]\\
&=\sum_{i=1}^{n}\sum_{k=1}^{c_i}\Delta_i(k).
\end{align}
An allocation $\vecc$ can be obtained by initially assigning the all zero vector $\veczero$ to $\vecc$ and then successively incrementing the entries of $\vecc$. After $M$ iterations, the constraint $\sum_ic_i=M$ is fulfilled and $\vecc$ is a valid allocation. If in some iteration, the $j$th entry is incremented by $1$, then the corresponding increment of the objective function is \text{$\Delta_j(c_j+1)$}. The following algorithm finds an allocation in a greedy manner. In each iteration, it increases by $1$ the entry $i$ with the smallest increment $\Delta_i(c_i+1)$.
\begin{algorithm}\label{alg:opt}\ 
\\
Initialize $c_i\leftarrow 0$, $i=1,\dotsc,n$.\\
\textbf{repeat} $M$ times\\
\indent Choose $\displaystyle j=\argmin_i \Delta_i(c_i+1)$.\\
\indent Update $c_j\leftarrow c_j+1$.\\
\textbf{end repeat}\\
Return $\vecc$.
\end{algorithm}
We next state the optimality of Algorithm~\ref{alg:opt}.
\begin{proposition}\label{prop:opt}
For a specified target pmf $\vect$ and a positive integer number $M$, the allocation $\vecc$ found by Algorithm~\ref{alg:opt} is a solution of Problem~\eqref{prob:mod}.
\end{proposition}
The proof is given in the next subsection.
\subsection{Proof of Proposition~\ref{prop:opt}}
We need the following two lemmas.
\begin{lemma}\label{lem:increasing}
For each $i$, if $k>\ell$ then $\Delta_i(k)>\Delta_i(\ell)$, i.e., the increment functions are strictly monotonically increasing.
\end{lemma}
\begin{proof}
We interpret the increment function $\Delta_i$ as defined on the set of real numbers greater than $1$ and calculate
\begin{align}
\frac{\partial}{\partial x}\Delta_i(x)=\log\frac{x}{x-1}>0.
\end{align}
\end{proof}
\begin{lemma}\label{lem:bounded}
Let $\vecc^*$ be an optimal allocation. Let $\vecc$ be a pre-allocation with $\sum_ic_i<M$ and $c_i\leq c^*_i$ for $i=1,\dotsc,n$. Define
\begin{align}
j=\argmin_i\Delta_i(c_i+1).\label{eq:defj}
\end{align}
Then for some optimal allocation $\tilde{\vecc}$ we have
\begin{align}
c_j+1&\leq \tilde{c}_j\label{eq:lem2:statement}\\
c_i&\leq\tilde{c}_i,\quad i=1,\dotsc,n.
\end{align}
\end{lemma}
\begin{proof}
Suppose we have
\begin{align}
c_j+1>c^*_j.\label{eq:lem2:supp}
\end{align}
Since $c_j\leq c^*_j$ by assumption, \eqref{eq:lem2:supp} implies
\begin{align}
c_j+1=c^*_j+1.\label{eq:lem2:prop1}
\end{align}
Since $\sum_ic_i<M$ and $\sum_ic^*_i=M$, there must be at least one $\ell\neq j$ with
\begin{align}
c^*_\ell\geq c_\ell+1.\label{eq:lem2:prop2}
\end{align}
By decreasing $c^*_\ell$ by one and increasing $c^*_j$ by one, the change of the objective function is $\Delta_j(c^*_j+1)-\Delta_\ell(c^*_\ell)$. We bound this change as follows.
\begin{align}
\Delta_j(c^*_j+1)-\Delta_\ell(c^*_\ell)&\overset{\text{(a)}}{\leq} \Delta_j(c^*_j+1)-\Delta_\ell(c_\ell+1)\label{eq:lem1:neq1}\\
&\overset{\text{(b)}}{=} \Delta_j(c_j+1)-\Delta_\ell(c_\ell+1)\\
&\overset{\text{(c)}}{\leq} \label{eq:lem1:neq2}0
\end{align}
where (a) follows by \eqref{eq:lem2:prop2} and Lemma~\ref{lem:increasing}, (b) follows by \eqref{eq:lem2:prop1}, and (c) follows by the definition of $j$ in \eqref{eq:defj}.
We have to consider two cases. First, suppose we have strict inequality in either \eqref{eq:lem1:neq1} or \eqref{eq:lem1:neq2}. Then the objective function is decreased, which contradicts the assumption that $\vecc^*$ is optimal. Thus, the supposition \eqref{eq:lem2:supp} is false and the statements of the lemma hold for $\tilde{\vecc}=\vecc^*$. Second, suppose we have equality both in \eqref{eq:lem1:neq1} and \eqref{eq:lem1:neq2}. In this case, define the allocation
\begin{align}
\tilde{c}_\ell = c^*_\ell-1,\quad\, \tilde{c}_j = c^*_j + 1,\quad\, \tilde{c}_i = c^*_i\text{ for } i\neq j,\ell.
\end{align}
Equality in \eqref{eq:lem1:neq1}--\eqref{eq:lem1:neq2} implies optimality of $\tilde{\vecc}$. By \eqref{eq:lem2:prop1} and \eqref{eq:lem2:prop2}, we can verify that $\tilde{\vecc}$ fulfills the statements of the lemma. This concludes the proof.
\end{proof}
By Lemma~\ref{lem:bounded}, there is an optimal allocation $\tilde{\vecc}$ such that in each iteration of Algorithm~\ref{alg:opt} we have 
\begin{align}
c_i\leq \tilde{c}_i,\qquad i=1,\dotsc,n.\label{eq:proof1}
\end{align}
After termination of Algorithm~\ref{alg:opt}, we have
\begin{align}
M=\sum_i c_i \leq \sum_i \tilde{c}_i = M.\label{eq:proof2}
\end{align}
Statements \eqref{eq:proof1} and \eqref{eq:proof2} can only be true simultaneously if $c_i= \tilde{c}_i$ for all $i=1,\dotsc,n$. Consequently, the constructed allocation $\vecc$ is optimal.	 This concludes the proof of Prop.~\ref{prop:opt}.
\subsection{Complexity}
Algorithm~\ref{alg:opt} must find the minimum of a vector with $n$ elements in each iteration, which is of complexity $\mathcal{O}(n)$. The algorithm terminates after $M$ iterations, so the overall complexity is $\mathcal{O}(nM)$. The complexity could be further reduced to $\mathcal{O}(M\log n)$ by keeping the list of increments $\Delta_i(c_i+1)$ sorted, but the presented algorithm is simple to implement and fast enough for our numerical calculations.

\subsection{Summary}
We summarize the properties found for $M$-type approximations of a target pmf $\vect$ in the following proposition. Note that the result of Prop.~\ref{prop:rate} for $M$-type approximations by quantization carries over to optimal $M$-type approximations.
\begin{proposition}
Let $\vecd_M$ be a pmf that minimizes $\kl(\vecd\Vert\vect)$ over all $M$-type pmfs. Then
\begin{enumerate}
\item $\kl(\vecd_M\Vert\vect)\leq T/M$, where $T>0$ depends on $\vect$ but not on $M$.
\item $\lim_{M\to\infty}\kl(\vecd_M\Vert\vect)\leq \lim_{M\to\infty}T/M=0$.
\item Algorithm~\ref{alg:opt} finds a $\vecd_M$ with a complexity of $\mathcal{O}(Mn)$.
\end{enumerate}
\end{proposition}

%% file: awgn.tex
\section{Overview: Approaching AWGN Capacity}
\label{sec:awgn}
We now consider the problem of approaching AWGN capacity. We briefly review existing results. 

Consider an AWGN channel with noise $N\sim\mathcal{N}(0,1)$. The channel capacity is (see \cite[Sec.~7.4]{Gallager1968})
\begin{align}
\mathsf{C}(\mathsf{snr})\!:=\frac{1}{2}\log(1+\mathsf{snr}).
\end{align}
Suppose we use polar coding with a discrete interface with $2^m$ points \cite[Sec. IV.A]{abbe2011polar}. We model this interface by an auxiliary random vector $\boldsymbol{Z}_m$ with $m$ binary entries $Z_i$ that are independent and uniformly distributed. Consequently, $\boldsymbol{Z}_m$ is uniformly distributed over
\begin{align}
\mathcal{Z}_m = \{\underbrace{0\dotsb 0}_{m\text{ bits}},0\dotsb01,\dotsc,1\dotsb 1\}. 
\end{align}
Consider a discrete set $\mathcal{X}_n$ of $|\mathcal{X}_n|=n$ real valued signal points and a deterministic mapping
\begin{align}
g\colon \mathcal{Z}_m\to\mathcal{X}_n.
\end{align}
The constellation $\mathcal{X}_n$ and the mapping $g$ are subject to the constraint
\begin{align}
\expop[g(Z_m)^2]\leq 1.
\end{align} 
Define the gap to capacity as
\begin{align}
D_m(\mathsf{snr},\mathcal{X}_n,g)\!:=\!\mathsf{C}(\mathsf{snr})\!-\!\mathbb{I}[g(Z_m);g(Z_m)\sqrt{\mathsf{snr}}+N]\label{eq:gap}
\end{align}
where $\mathbb{I}(X;Y)$ is the mutual information between $X$ and $Y$. We would like to know how the gap \eqref{eq:gap} scales with the number $m$ of bits at the uniform interface. Two special cases are of interest: First, when $n=2^m$ and the mapping $g$ is one-to-one. In this case, the signal point pmf is uniform and optimization is only over the signal point \emph{positions} $\mathcal{X}_n$. This approach is called \emph{geometric shaping}. Second, the signal point positions $\mathcal{X}_n$ are restricted to be equidistant with distance $\Delta$. In this case, optimization is over the distance $\Delta$, the number of signal points $n$, and the mapping $g$. This approach is called \emph{probabilistic shaping}. 
\subsection{Previous Result: Geometric Shaping}
Abbe and Barron show in \cite[Sec. IV.C]{abbe2011polar} the existence of a family $\mathcal{X}_n$ such that for $n=2^m$ and $g$ being one-to-one (we indicate this by writing $g_\mathrm{id}$), the gap to capacity scales as
\begin{align}
D_m(\mathsf{snr},\mathcal{X}_{2^m},g_\mathrm{id})\leq \snr \cdot 2^{-m}.\label{eq:geometric}
\end{align}
In other words, there exist signal point constellations $\mathcal{X}_{2^m}$ such that the gap to capacity decreases at least exponentially in the number of bits $m$ at the uniform interface when the mapping $g$ is one-to-one. Note that the constellations $\mathcal{X}_{2^m}$ that achieve this behavior are not equidistant.
\subsection{Previous Result: Probabilistic Shaping}
Abbe and Telatar propose in \cite[Sec. V]{Abbe2010} to use $m+1$ equidistant signal points and binomial coefficients normalized by $2^m$ as a $2^m$-type pmf over these points. They call this scheme the CLT approach. We denote the equidistant signal points by $\mathcal{E}_{m+1}$ and the mapping defined by the binomial coefficients by $g_\mathrm{clt}$. Abbe and Barron show in \cite[Sec. IV.B]{abbe2011polar} that 
\begin{align}
D_m(\mathsf{snr},\mathcal{E}_{m+1},g_\mathrm{clt})\leq B_{\snr} \cdot m^{-1}\label{eq:probabilistic}
\end{align}
for some constant $B_{\snr}>0$ that depends on the $\snr$. The bound \eqref{eq:probabilistic} implies that with the CLT approach the capacity gap decreases at least as $m^{-1}$ in the number of bits at the uniform interface. Comparing \eqref{eq:geometric} and \eqref{eq:probabilistic}, we see that geometric shaping outperforms the CLT approach. This motivates improving the CLT approach.

%% file: numerical.tex
\section{Improved Non-Uniform Mapping for AWGN}
\label{sec:numerical}
\begin{figure*}[ht!]
\footnotesize{(a) 0dB. The horizontal and vertical axis display signal point position and probability$\times 2^m$, respectively.}\\
\includegraphics[width=\textwidth]{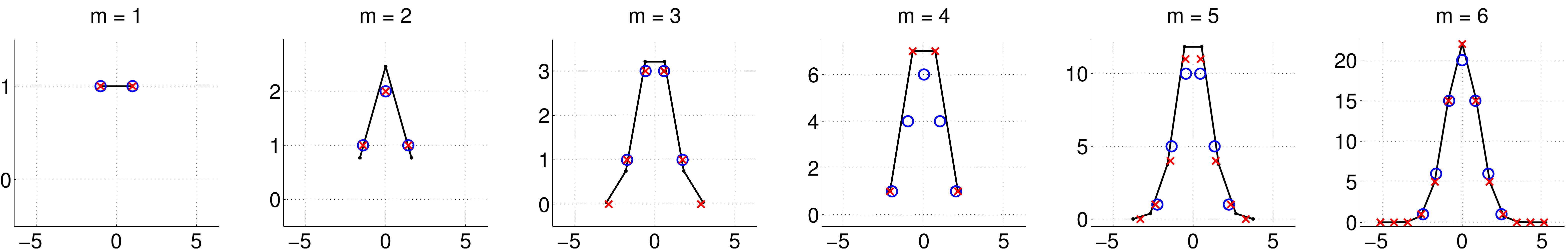}\\[0.2cm]
\footnotesize{(b) 5dB. The horizontal and vertical axis display signal point position and probability$\times 2^m$, respectively.}\\
\includegraphics[width=\textwidth]{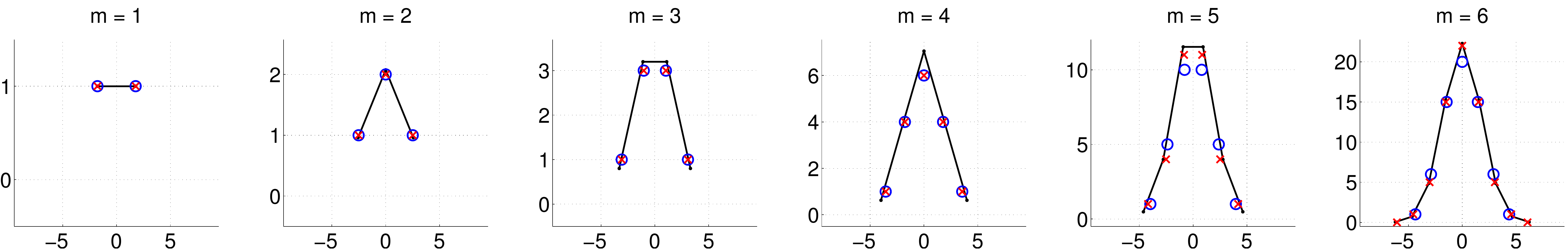}\\[0.4cm]
\def\svgwidth{1.0\columnwidth}
\parbox{\columnwidth}{(c) 0dB\\[-0.5cm]%
\executeiffilenewer{figures/M-gap-0dB.svg}{figures/M-gap-0dB.pdf}%
{inkscape -z -D --file=figures/M-gap-0dB.svg %
--export-pdf=figures/M-gap-0dB.pdf --export-latex}%
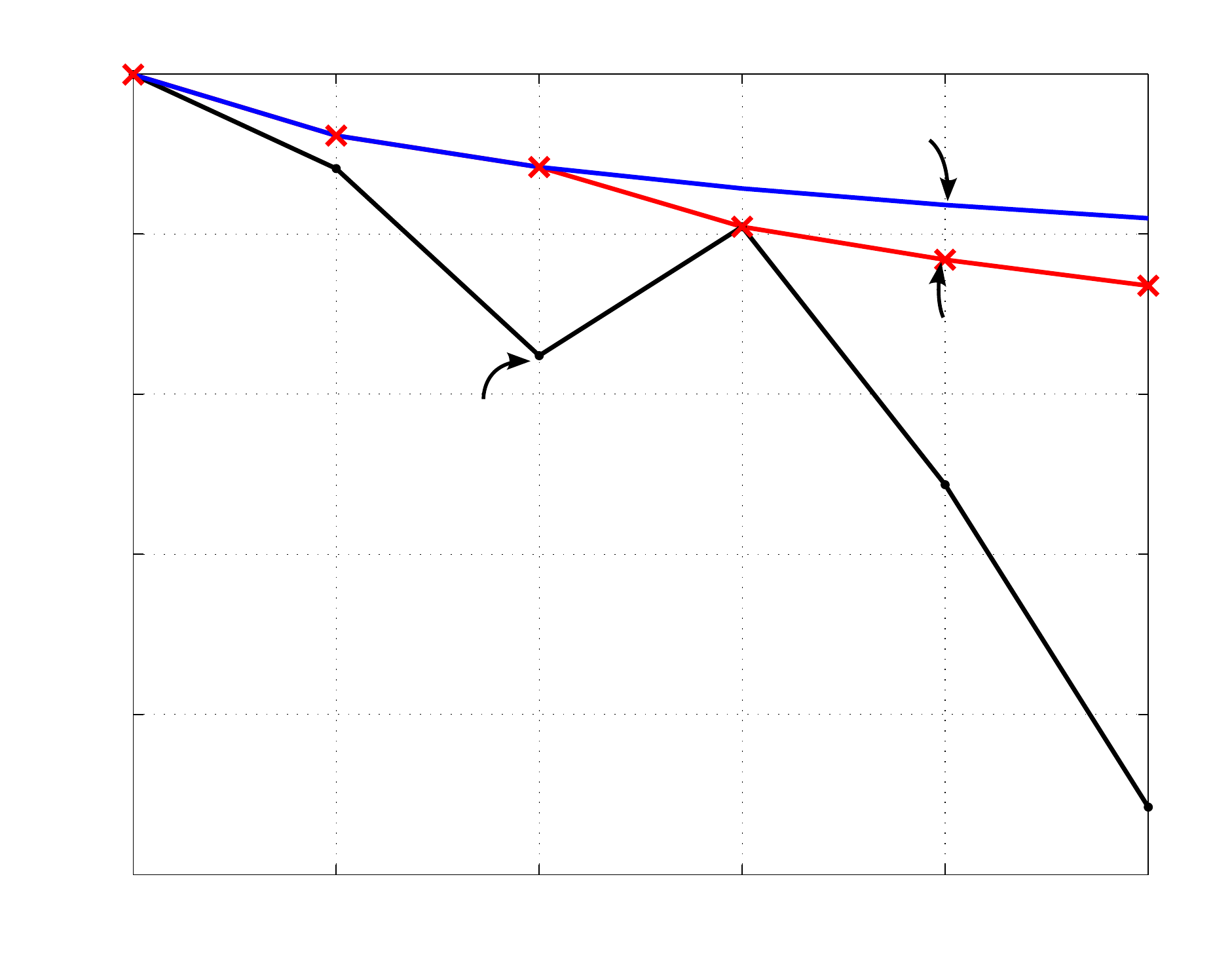%
}
\def\svgwidth{1.0\columnwidth}
\parbox{\columnwidth}{(d) 5dB\\[-0.5cm]%
\executeiffilenewer{figures/M-gap-5dB.svg}{figures/M-gap-5dB.pdf}%
{inkscape -z -D --file=figures/M-gap-5dB.svg %
--export-pdf=figures/M-gap-5dB.pdf --export-latex}%
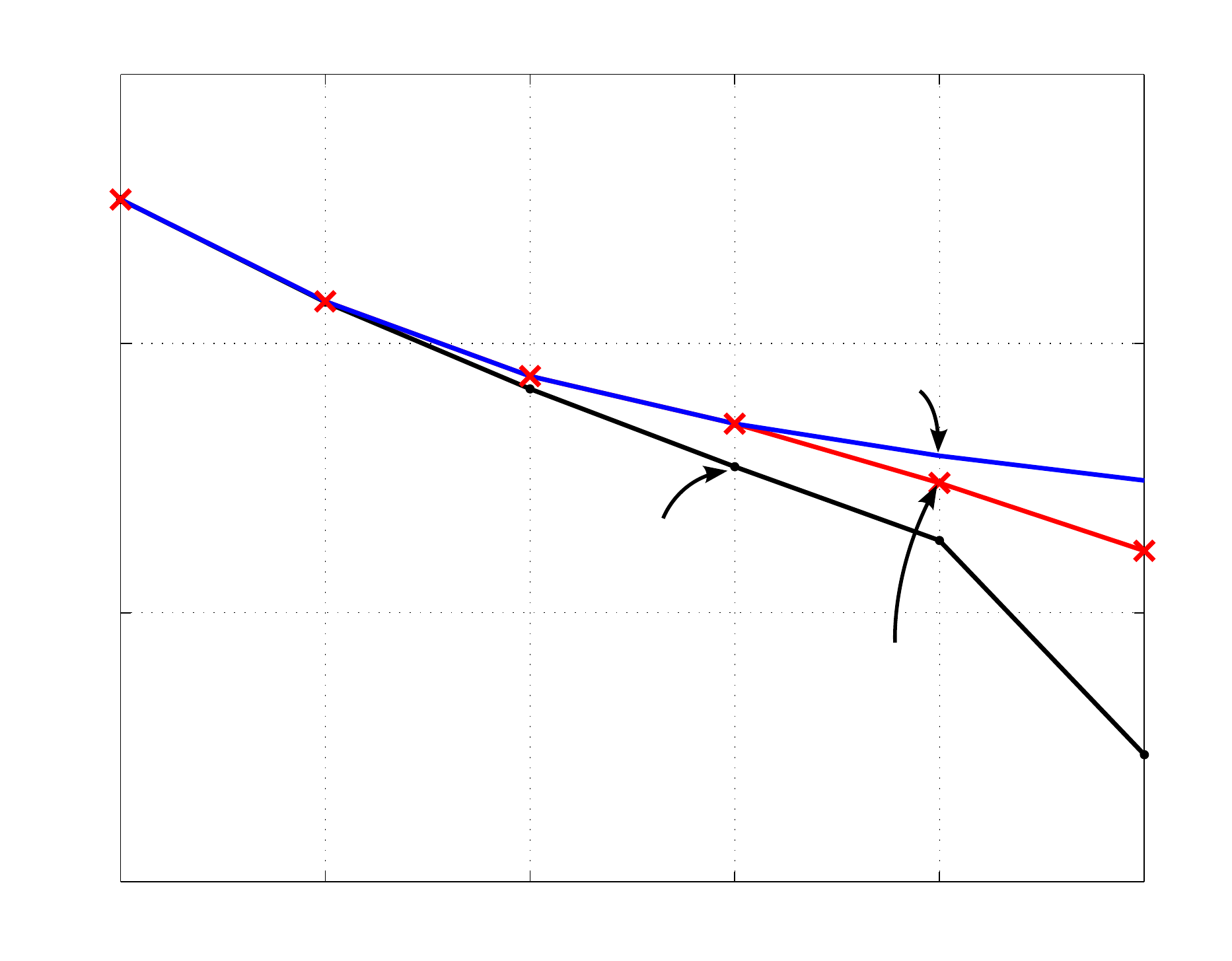%
}
\caption{Comparison of CLT approach with optimal non-uniform mapping as proposed in this work.}
\label{fig:gap}
\end{figure*}
The key observation is as follows. For a given $m$, the CLT approach provides $m+1$ constellation points and a fixed pmf over these points independent of the $\snr$. This approach achieves capacity for any value of the $\snr$ for $m\to\infty$. Intuitively this approach should be sub-optimal in general for finite values of $m$. This can be seen as follows. For a fixed $m$ and high enough $\snr$, we expect among all $2^m$-type pmfs the uniform pmf over $2^m$ points to be optimal. However, the CLT approach limits the number of constellation points to $m+1$. We therefore propose to maximize both over the \emph{cardinality} of the constellation and the \emph{pmf}. Note that there is a tradeoff between the constellation size and the pmf resolution. If we have $n$ constellation points, we have a resolution of $2^m/n$ on average for the probability of each constellation point.
\subsection{Our Approach}
\begin{figure}
\begin{algorithm}\label{alg:awgn}\ 
\\
\textbf{for} $k=2,\dotsc,2^m$\\
\indent 1. $\mathcal{X}^{(k)}:=$ $k$ points: equidistant, normalized, centered.\\
\indent 2. solve
\begin{align*}
\maximize_{\vecp,\Delta}\quad&\miop(X\sqrt{\snr};Y)\\
\st\quad &X\sim\vecp,\,X\in\Delta\mathcal{X}^{(k)},\,\expop(|\Delta X|^2)\leq 1.
\end{align*}
\indent \phantom{3.} Denote optimal pmf by $\vecp^*$.\\
\indent 3. $\vecd^{(k)}:=$ $2^m$-type pmf that minimizes $\kl(\vecd\Vert\vecp^*)$.\\
\textbf{end for}\\
4. Choose $n=\displaystyle\argmin_{k}\miop(\vecd^{(k)})$.\\
\end{algorithm}
\end{figure}
In Alg.~\ref{alg:awgn}, we state our approach as an algorithm. We next give details for each step.

\emph{Step 1.} Self-explanatory.

\emph{Step 2.} We calculate the capacity-achieving pmf of a constellation that consists of $k$ equidistant points. The optimization is both over the distance $\Delta$ of the points and over the input pmf. The optimization over $\Delta$ is done by line search and for each $\Delta$ the optimization over $\vecp$ is a convex optimization problem. We let $\Delta$ take a finite number of equally spaced values, and for each value we solve the convex optimization problem by using CVX \cite{CVX}. We then choose $\vecp^*$ as the optimal pmf for the value of $\Delta$ that results in the greatest mutual information.

\emph{Step 3.} For the optimal pmf $\vecp^*$ that we found in step 2., we use Algorithm~\ref{alg:opt} to find the pmf that minimizes $\kl(\vecd\Vert\vecp^*)$ over all $2^m$-type pmfs $\vecd$. Note that by \cite[Prop. 5.10]{bocherer2012capacity}, \cite[Prop. 3.11]{bocherer2012capacity}, and Pinsker's inequality \cite[Theorem 1.5]{kramer2012multi}, if $\kl(\vecd\Vert\vecp^*)\rightarrow 0$ then the mutual information and the average power achieved by $\vecd$ converge respectively to the mutual information and the average power achieved by $\vecp^*$. To avoid unfair comparisons, we guarantee that the power constraint is fulfilled with equality by rescaling the constellation appropriately, i.e., we calculate the distance $\Delta^{(k)}$ by
\begin{align}
\expop_{\vecd^{(k)}}(|\Delta X|^2)\overset{!}{=}1\quad\Rightarrow\quad\Delta^{(k)}=\frac{1}{\sqrt{\expop_{\vecd^{(k)}}(|X|^2)}}.
\end{align}

\emph{Step 4.} For each constellation size $2,\dotsc,2^m$, the algorithm calculates a $2^m$-type pmf. Choose the one that yields the greatest mutual information.
\subsection{Numerical Results}
We apply Algorithm~\ref{alg:awgn} for signal-to-noise ratios of 0dB and 5dB, i.e., the $\snr$ takes the values $1$ and $\approx 3.16$, respectively. We let $m$ take the values $1,2,3,4,5,6$. Fig.~\ref{fig:gap} (a) and (c) show the results for $0$dB and Fig.~\ref{fig:gap} (b) and (d) show the results for $5$dB. We discuss only the results for $0$dB, the results for $5$dB are similar.

 For each value of $m$, we display in Fig.~\ref{fig:gap} (a) the results for the CLT approach by a blue circle. The horizontal coordinate represents the position of a signal point and the vertical coordinate its probability scaled by the factor $2^m$. The black points connected by a line represent the target pmf $\vecp^{*(n)}$ and the red cross represents its $2^m$-type approximation $\vecd^{(n)}$ as chosen by Algorithm~\ref{alg:awgn} in line 4. As can be seen, for $m=1,2,3$, Algorithm~\ref{alg:awgn} recovers the $2^m$-type pmf obtained via the CLT approach. For $m=4,5,6$, the $2^m$-type pmfs chosen by Algorithm~\ref{alg:awgn} differ from the CLT pmfs.

 It is important to note that Algorithm~\ref{alg:awgn} chooses a different number of signal points than the CLT approach. In Fig.~\ref{fig:gap} (c) the gap to capacity in nats is displayed. The blue line indicates the gap achieved by the CLT approach. The curve appears logarithmic in the logarithmic scale, which is consistent with the behavior $1/m$ as predicted by \eqref{eq:probabilistic}. The black connected points indicate the gap that the target pmfs would achieve. Note that the gap is not monotonically decreasing in $m$. The reason for this is that Algorithm~\ref{alg:awgn} chooses in step 4. the target pmf $\vecp^{*(n)}$ according to the gap that is achieved by its $2^m$-type approximation $\vecd^{(n)}$, and not according to the gap that the target pmf would achieve by itself. 
\newpage

The gap achieved by the $2^m$-type approximation of the target pmfs is displayed by connected red crosses. Note that this gap actually decreases monotonically with $m$. As expected from Fig.~\ref{fig:gap} (a), the gaps achieved by CLT and Algorithm~\ref{alg:awgn} are identical for $m=1,2,3$. For $m=4,5,6$, our approach outperforms the CLT approach. Note that this smaller gap is achieved by using a different number of signal points than the CLT approach suggests. This shows that our idea of optimizing both over the probabilities and the number of signal points is beneficial. 
\subsection{Conclusions}
The numerical results suggest to look beyond the CLT approach and search for new analytical bounds for the gap that can be achieved by probabilistic shaping, i.e., equidistant constellations with non-uniform mappings. It may be possible that the scaling of geometric shaping \eqref{eq:geometric} can also be achieved by probabilistic shaping. This would be an interesting property, since geometrically shaped constellations need quantizers at the receiver of much higher precision than  equidistant constellations do. This makes the probabilistic shaping approach attractive for practical systems.

%% file: figures/M-gap-0dB.pdf_tex

\begingroup
  \makeatletter
  \providecommand\color[2][]{%
    \errmessage{(Inkscape) Color is used for the text in Inkscape, but the package 'color.sty' is not loaded}
    \renewcommand\color[2][]{}%
  }
  \providecommand\transparent[1]{%
    \errmessage{(Inkscape) Transparency is used (non-zero) for the text in Inkscape, but the package 'transparent.sty' is not loaded}
    \renewcommand\transparent[1]{}%
  }
  \providecommand\rotatebox[2]{#2}
  \ifx\svgwidth\undefined
    \setlength{\unitlength}{541.70935125pt}
  \else
    \setlength{\unitlength}{\svgwidth}
  \fi
  \global\let\svgwidth\undefined
  \makeatother
  \begin{picture}(1,0.77346709)%
    \put(0,0){\includegraphics[width=\unitlength]{M-gap-0dB.pdf}}%
    \put(0.09998656,0.03082447){\color[rgb]{0,0,0}\makebox(0,0)[lb]{\smash{1}}}%
    \put(0.26474176,0.03082447){\color[rgb]{0,0,0}\makebox(0,0)[lb]{\smash{2}}}%
    \put(0.42949695,0.03082447){\color[rgb]{0,0,0}\makebox(0,0)[lb]{\smash{3}}}%
    \put(0.59425214,0.03082447){\color[rgb]{0,0,0}\makebox(0,0)[lb]{\smash{4}}}%
    \put(0.75900733,0.03082447){\color[rgb]{0,0,0}\makebox(0,0)[lb]{\smash{5}}}%
    \put(0.92391019,0.03082447){\color[rgb]{0,0,0}\makebox(0,0)[lb]{\smash{6}}}%
    \put(0.03660796,0.05251491){\color[rgb]{0,0,0}\makebox(0,0)[lb]{\smash{10}}}%
    \put(0.06937439,0.07082719){\color[rgb]{0,0,0}\makebox(0,0)[lb]{\smash{-12}}}%
    \put(0.03660796,0.18265765){\color[rgb]{0,0,0}\makebox(0,0)[lb]{\smash{10}}}%
    \put(0.06937439,0.20096993){\color[rgb]{0,0,0}\makebox(0,0)[lb]{\smash{-10}}}%
    \put(0.03660796,0.31265273){\color[rgb]{0,0,0}\makebox(0,0)[lb]{\smash{10}}}%
    \put(0.06937439,0.33094657){\color[rgb]{0,0,0}\makebox(0,0)[lb]{\smash{-8}}}%
    \put(0.03660796,0.44262936){\color[rgb]{0,0,0}\makebox(0,0)[lb]{\smash{10}}}%
    \put(0.06937439,0.46094164){\color[rgb]{0,0,0}\makebox(0,0)[lb]{\smash{-6}}}%
    \put(0.03660796,0.57262444){\color[rgb]{0,0,0}\makebox(0,0)[lb]{\smash{10}}}%
    \put(0.06937439,0.59093672){\color[rgb]{0,0,0}\makebox(0,0)[lb]{\smash{-4}}}%
    \put(0.03660796,0.70261951){\color[rgb]{0,0,0}\makebox(0,0)[lb]{\smash{10}}}%
    \put(0.06937439,0.72091333){\color[rgb]{0,0,0}\makebox(0,0)[lb]{\smash{-2}}}%
    \put(0.25551175,0.00067933){\color[rgb]{0,0,0}\makebox(0,0)[lb]{\smash{number m of bits at the uniform interface}}}%
    \put(0.0215317,0.24141572){\color[rgb]{0,0,0}\rotatebox{90}{\makebox(0,0)[lb]{\smash{gap to capacity in nats}}}}%
    \put(0.15602009,0.42665058){\color[rgb]{0,0,0}\makebox(0,0)[lb]{\smash{gap achieved by $\vecp^*$ as chosen }}}%
    \put(0.1554667,0.39537456){\color[rgb]{0,0,0}\makebox(0,0)[lb]{\smash{in step 2. of Algorithm 2}}}%
    \put(0.50324396,0.6722201){\color[rgb]{0,0,0}\makebox(0,0)[lb]{\smash{gap achieved by CLT approach }}}%
    \put(0.71858429,0.4951645){\color[rgb]{0,0,0}\makebox(0,0)[lb]{\smash{gap achieved by }}}%
    \put(0.72140092,0.466627){\color[rgb]{0,0,0}\makebox(0,0)[lb]{\smash{Algorithm 2}}}%
  \end{picture}%
\endgroup

%% file: figures/M-gap-5dB.pdf_tex

\begingroup
  \makeatletter
  \providecommand\color[2][]{%
    \errmessage{(Inkscape) Color is used for the text in Inkscape, but the package 'color.sty' is not loaded}
    \renewcommand\color[2][]{}%
  }
  \providecommand\transparent[1]{%
    \errmessage{(Inkscape) Transparency is used (non-zero) for the text in Inkscape, but the package 'transparent.sty' is not loaded}
    \renewcommand\transparent[1]{}%
  }
  \providecommand\rotatebox[2]{#2}
  \ifx\svgwidth\undefined
    \setlength{\unitlength}{537.10647523pt}
  \else
    \setlength{\unitlength}{\svgwidth}
  \fi
  \global\let\svgwidth\undefined
  \makeatother
  \begin{picture}(1,0.7800115)%
    \put(0,0){\includegraphics[width=\unitlength]{M-gap-5dB.pdf}}%
    \put(0.08967252,0.03108863){\color[rgb]{0,0,0}\makebox(0,0)[lb]{\smash{1}}}%
    \put(0.25583964,0.03108863){\color[rgb]{0,0,0}\makebox(0,0)[lb]{\smash{2}}}%
    \put(0.42200674,0.03108863){\color[rgb]{0,0,0}\makebox(0,0)[lb]{\smash{3}}}%
    \put(0.58817385,0.03108863){\color[rgb]{0,0,0}\makebox(0,0)[lb]{\smash{4}}}%
    \put(0.75434095,0.03108863){\color[rgb]{0,0,0}\makebox(0,0)[lb]{\smash{5}}}%
    \put(0.92065699,0.03108863){\color[rgb]{0,0,0}\makebox(0,0)[lb]{\smash{6}}}%
    \put(0.03692168,0.05296495){\color[rgb]{0,0,0}\makebox(0,0)[lb]{\smash{10}}}%
    \put(0.06996891,0.07143416){\color[rgb]{0,0,0}\makebox(0,0)[lb]{\smash{-6}}}%
    \put(0.03692168,0.27157939){\color[rgb]{0,0,0}\makebox(0,0)[lb]{\smash{10}}}%
    \put(0.06996891,0.29003){\color[rgb]{0,0,0}\makebox(0,0)[lb]{\smash{-4}}}%
    \put(0.03692168,0.49017526){\color[rgb]{0,0,0}\makebox(0,0)[lb]{\smash{10}}}%
    \put(0.06996891,0.50864447){\color[rgb]{0,0,0}\makebox(0,0)[lb]{\smash{-2}}}%
    \put(0.24653053,0.00068515){\color[rgb]{0,0,0}\makebox(0,0)[lb]{\smash{number m of bits at the uniform interface}}}%
    \put(0.02171622,0.2434846){\color[rgb]{0,0,0}\rotatebox{90}{\makebox(0,0)[lb]{\smash{gap to capacity in nats}}}}%
    \put(0.29981178,0.33615065){\color[rgb]{0,0,0}\makebox(0,0)[lb]{\smash{gap achieved by $\vecp^*$ as chosen }}}%
    \put(0.29925362,0.30312168){\color[rgb]{0,0,0}\makebox(0,0)[lb]{\smash{in step 2. of Algorithm 2}}}%
    \put(0.4931837,0.47524029){\color[rgb]{0,0,0}\makebox(0,0)[lb]{\smash{gap achieved by CLT approach }}}%
    \put(0.62714431,0.23404899){\color[rgb]{0,0,0}\makebox(0,0)[lb]{\smash{gap achieved by }}}%
    \put(0.62998508,0.20526693){\color[rgb]{0,0,0}\makebox(0,0)[lb]{\smash{Algorithm 2}}}%
  \end{picture}%
\endgroup